\documentclass[11pt,leqno]{article}

\input{style}  
\begin{document}
%%%%%%%%%%%%%%%%%%%%%%%%%%%%%%%%%%%%%%%%%%%%%%%%%%%%%%%%%%%%%%%%%%

\title{On a stochastic differential equation arising in a price impact
  model}
 
\author{P. Bank\\ Technische
  Universit{\"a}t Berlin\\
  Institut f{\"u}r Mathematik\\ Stra{\ss}e des 17. Juni 136, 10623
  Berlin, Germany \\ (bank@math.tu-berlin.de) \and D. Kramkov
  \thanks{This research was supported in part by the Carnegie
    Mellon-Portugal Program and by the Oxford-Man Institute for
    Quantitative Finance at the University of Oxford.}  \\
  Carnegie Mellon University and University of Oxford,\\ Department of
  Mathematical
  Sciences,\\ 5000 Forbes Avenue, Pittsburgh, PA, 15213-3890, US \\
  (kramkov@cmu.edu)} \date{\today}

\maketitle
\begin{abstract}
  We provide sufficient conditions for the existence and uniqueness of
  solutions to a stochastic differential equation which arises in the
  price impact model developed in \cite{BankKram:11a} and
  \cite{BankKram:11b}. These conditions are stated as smoothness and
  boundedness requirements on utility functions or Malliavin
  differentiability of payoffs and endowments.
\end{abstract}
 
\begin{description}
\item[Keywords:] Clark-Ocone formula, large investor, Malliavin
  derivative, Pareto allocation, price impact, Sobolev embedding,
  stochastic differential equation.
\item[JEL Classification:] G11, G12, C61.
\item[AMS Subject Classification (1991):] 90A09, 90A10, 90C26.
\end{description}
  
\section{Introduction}
\label{sec:introduction}

In~\cite{BankKram:11a, BankKram:11b}, we developed a
financial model for a large investor who trades with market makers at
their utility indifference prices. We showed that the evolution of
this system can be described by a nonlinear stochastic differential
equation; see~\eqref{eq:10}.

It is the purpose of this paper to derive conditions for existence and
uniqueness of solutions to this SDE. A special feature of our study is
that the SDE's coefficients are defined only implicitly and, hence,
standard Lipschitz and growth conditions are not easily applicable.
We aim to provide readily verifiable criteria in terms of the model
primitives: the market makers' utility functions and initial
endowments and stocks' dividends.

Our main results, stated in Section~\ref{sec:main-results}, yield such
conditions for locally bounded order flows. Theorem~\ref{th:1} shows
that if the market makers' risk aversions are bounded along with
sufficiently many of their derivatives then there exist unique maximal
local solutions.  Its proof relies on Sobolev's embedding results for
stochastic integrals due to \citet{Szn:81}. For the special case of
exponential utilities Theorem~\ref{th:2} establishes the existence of
a unique global solution.  Theorem~\ref{th:3} proves this under the
alternative assumption that, in a Brownian framework, the market
makers' initial endowments and stocks' dividends are Malliavin
differentiable and risk aversions are bounded along with their first
derivatives. The main tool here is the Clark-Ocone formula for
$\mathbf{D}^{1,1}$ from \citet*{KaratOconeLi:91}.

\subsection*{Some notation.}
\label{sec:some-notation}

We use the conventions and notations of the parent paper
\cite[Section~2]{BankKram:11b}.  In particular, for a metric space
$\mathbf{X}$ we denote by $\mathbf{C}([0,1],\mathbf{X})$ the space of
continuous maps from $[0,1]$ to $\mathbf{X}$.  For nonnegative
integers $m$ and $n$ and an open set $V\subset \mathbf{R}^d$ we denote
by $\mathbf{C}^m = \mathbf{C}^m(V) = \mathbf{C}^m(V, \mathbf{R}^n)$
the Fr\'echet space of $m$-times continuously differentiable functions
$\map{f}{V}{\mathbf{R}^n}$ with the topology generated by the
semi-norms
\begin{displaymath}
  \norm{f}_{m,C} \set \sum_{0\leq \abs{\beta} \leq m} \sup_{x\in C}
  \abs{\partial^\beta f(x)}, 
\end{displaymath}
where $C$ is a compact subset of $V$, $\beta =
(\beta_1,\dots,\beta_d)$ is a multi-index of non-negative integers,
$\abs{\beta} \set \sum_{i=1}^d \beta_i$, and
\begin{displaymath}
  \partial^\beta \set 
  \frac{\partial^{|\beta|}}{\partial x_1^{\beta_1}\dots \partial 
    x_d^{\beta_d}}. 
\end{displaymath}
If $K = (K(x))_{x\in V}$ is a family of stochastic processes $K(x) =
(K_t(x))_{t\in [0,1]}$, then we say that $K$ has values in
$\mathbf{C}^m(V)$ if for every $t\in [0,1]$ the stochastic field $K_t$
on $V$ has sample paths in $\mathbf{C}^m(V)$.

\section{Setup}
\label{sec:setup}

Let $u_m = u_m(x)$, $m=1,\dots,M$, be functions on the real line
$\mathbf{R}$ satisfying

\begin{Assumption}
  \label{as:1}
  Each $u_m$ is strictly concave, strictly increasing, twice
  continuously differentiable,
  \begin{displaymath}
    \lim_{x\to \infty} u_m(x) = 0
  \end{displaymath}
  and for some constant $c>0$ the absolute risk aversion
  \begin{equation}
    \label{eq:1}
    \frac1c \leq a_m(x) \triangleq -\frac{u_m''(x)}{u_m'(x)} \leq c,
    \quad x \in \mathbf{R}.
  \end{equation}
\end{Assumption}

Denote by $r=r(v,x)$ the $v$-weighted $\sup$-convolution:
\begin{equation}
  \label{eq:2}
  r(v,x) \set \max_{x^1+\dots+ x^M = x} \sum_{m=1}^M v^m u_m(x^m),
  \quad (v,x) \in (0,\infty)^M \times \mathbf{R}.
\end{equation}
The main properties of the saddle function $r=r(v,x)$ are collected in
\cite[Section 4.1]{BankKram:11a}. In particular, for $v\in
(0,\infty)^M$, the function $r(v,\cdot)$ has the same properties as
the functions $u_m$, $m=1,\dots,M$, of Assumption~\ref{as:1} and, for
$c>0$ from~\eqref{eq:1},
\begin{align}
  \label{eq:3}
  \frac1c \frac{\partial r}{\partial x}(v,x) \leq -M\frac{\partial^2
    r}{\partial x^2}(v,x) \leq c \frac{\partial r}{\partial x}(v,x),
  \quad x \in \mathbf{R}.
\end{align}
From~\eqref{eq:3} we deduce the exponential growth property
\begin{equation}
  \label{eq:4}
  e^{-y^+c/M+y^-/(cM)} \leq \frac{\frac{\partial r}{\partial
      x}(v,x+y)}{\frac{\partial
      r}{\partial x}(v,x)} \leq e^{-y^+/(cM)+y^-c/M}, \quad x,y \in \mathbf{R}, 
\end{equation}
where for real $x$ we denote by $x^+\set \max(x,0)$ and $x^- \set
(-x)^+$ the positive and negative parts of $x$.  As $r(v,x)\to 0$ when
$x\to \infty$, we also obtain the estimates
\begin{align}
  \label{eq:5}
  -\frac1c r(v,x) \leq M\frac{\partial r}{\partial x}(v,x) \leq -c
  r(v,x), \quad x \in \mathbf{R}.
\end{align}

Let $(\Omega, \mathcal{F}_1, (\mathcal{F}_t)_{0 \leq t \leq 1},
\mathbb{P})$ be a complete filtered probability space satisfying
\begin{Assumption}
  \label{as:2}
  There is a $d$-dimensional Brownian motion $B = (B^i)$ such that
  every local martingale $M$ admits an integral representation
  \begin{displaymath}
    M_t = M_0 + \int_0^t H_u dB_u \set M_0 + \sum_{i=1}^d \int_0^t H^i_u
    dB^i_u, \quad t\in [0,1], 
  \end{displaymath}
  for some predictable process $H = (H^i)$ with values in
  $\mathbf{R}^d$.
\end{Assumption}
\noindent Of course, this assumption holds if the filtration is
generated by $B$.

Let $\Sigma_0$ and $\psi = (\psi^j)_{j=1,\dots,J}$ be random
variables. We denote
\begin{displaymath}
  \Sigma(x,q) \set \Sigma_0 + x+\ip{q}{\psi} \set \Sigma_0  + x +
  \sum_{j=1}^J q^j \psi^j, \quad (x,q) \in \mathbf{R} 
  \times \mathbf{R}^J,
\end{displaymath}
and assume that
\begin{equation}
  \label{eq:6}
  \mathbb{E}[r(v,\Sigma(x,q))] > -\infty, \quad 
  (v,x,q) \in \mathbf{A} \set (0,\infty)^M \times \mathbf{R} 
  \times \mathbf{R}^J.
\end{equation}
From~\eqref{eq:4} and~\eqref{eq:5} we deduce that this integrability
condition holds if
\begin{equation}
  \label{eq:7}
  \mathbb{E}[e^{p\abs{\psi} + c\Sigma_0^-/M}] < \infty, \quad p>0.
\end{equation}

Under Assumptions~\ref{as:1} and \ref{as:2} and the integrability
condition~\eqref{eq:6} the stochastic fields
\begin{align*}
  F_t(v,x,q) &\set \mathbb{E}[r(v,\Sigma(x,q))|\mathcal{F}_t], \quad
  (v,x,q)\in \mathbf{A}, \\
  \begin{split}
    G_t(u,y,q) &\set \sup_{v\in (0,\infty)^M}\inf_{x\in
      \mathbf{R}}[\ip{v}{u} +
    xy  - F_t(v,x,q)], \\
    &\quad (u,y,q) \in \mathbf{B}\set (-\infty,0)^M \times (0,\infty)
    \times \mathbf{R}^J,
  \end{split}
\end{align*}
are well-defined, have sample paths in
$\mathbf{C}([0,1],\mathbf{C}^2(\mathbf{A}))$ and
$\mathbf{C}([0,1],\mathbf{C}^2(\mathbf{B}))$, respectively, and for a
multi-index $\beta = (\beta_1,\dots,\beta_{M+1+J})$ with
$\abs{\beta}\leq 2$
\begin{displaymath}
  \partial^{\beta} F_t(v,x,q) =
  \mathbb{E}[\partial^\beta r(v,\Sigma(x,q))|\mathcal{F}_t];
\end{displaymath}
see Theorems~4.3 and~4.4 and Corollary~5.4 in \cite{BankKram:11b}.  In
view of Assumption~\ref{as:2} the martingales $\partial^\beta F$,
$\abs{\beta}\leq 2$, admit integral representations:
\begin{equation}
  \label{eq:8}
  \partial^{\beta} F_t(v,x,q) =  \mathbb{E}[\partial^\beta
  r(v,\Sigma(x,q))] + \int_0^t 
  \partial^{\beta} H_s(v,x,q)dB_s 
\end{equation}
for some $d$-dimensional predictable processes $\partial^{\beta}
H(v,x,q)$, where the notation $\partial^{\beta} H$ can be justified
using the concept of $\mathcal{L}$-derivatives from \cite[Section
2.7]{Krylov:80}. Note that under the conditions of our main
Theorems~\ref{th:1}--\ref{th:3} the derivatives
$\partial^\beta H_t$ are well-defined in the usual pointwise sense and
have continuous sample paths on $\mathbf{A}$; see Lemmas~\ref{lem:10}
and~\ref{lem:13}.

As part of the construction of our SDE~\eqref{eq:10} we assume that
the process $\frac{\partial H}{\partial v}$ has values in
$\mathbf{C}(\mathbf{A},\mathbf{R}^{d})$ (equivalently, that the
stochastic fields $\frac{\partial H^i_t}{\partial v^m}$ on
$\mathbf{A}$ are continuous) and, for $u \in (-\infty,0)^M$ and $q \in
\mathbf{R}^J$, define an $M\times d$-dimensional process $K(u,q)$ by
\begin{equation}
  \label{eq:9}
  K^{mi}_t(u,q) \set \frac{\partial H^i_t}{\partial
    v^m}\left(\frac{\partial G_t}{\partial 
      u}(u,1,q),G_t(u,1,q),q\right), \quad t\in [0,1].   
\end{equation}
The paper is concerned with the existence and uniqueness of a (strong)
solution $U = (U^m)_{m=1,\dots,M}$ with values in $(-\infty,0)^M$ to
the stochastic differential equation
\begin{equation}
  \label{eq:10}
  U_t = U_0 + \int_0^t K_s(U_s,Q_s) dB_s, \quad U_0 \in
  (-\infty,0)^M, 
\end{equation}
parameterized by a predictable process $Q$ with values in
$\mathbf{R}^J$.

This equation arises in the price impact model of \cite{BankKram:11b},
where it describes the evolution of the expected utilities $U = (U^m)$
of $M$ market makers who collectively acquire $Q = (Q^j)$ stocks from
a ``large'' investor. The functions $u_m = u_m(x)$, $m=1,\dots,M$,
specify the market makers' utilities for terminal wealth and
$\Sigma_0$ stands for their total initial random endowment.  The
cumulative dividends paid by the stocks are given by $\psi =
(\psi^j)$. According to \cite[Theorem 5.8]{BankKram:11b} a predictable
process $Q = (Q^j)$ is a (well-defined) investment strategy for the
large investor if and only if~\eqref{eq:10} has a unique (global)
solution $U$. In this case, the total cash amount received by the
market makers (and paid by the investor) up to time $t$ is given by
$G_t(U_t,1,Q_t)$.

Of course, it is easy to state standard conditions on the stochastic
field $K$ guaranteeing the existence and uniqueness of a solution $U$
to~\eqref{eq:10}; see Lemma~\ref{lem:2} below. However, such criteria
have little practical value, because, except in special cases such as
Example 5.9 in \cite{BankKram:11b}, an explicit expression for $K$ is
not available. Instead, we look for easily verifiable conditions in
terms of the model primitives: the functions $(u_m)$ and the random
variables $\Sigma_0$ and $\psi=(\psi^j)$.

\section{Main results}
\label{sec:main-results}

Let $Q$ be a predictable process with values in $\mathbf{R}^J$ and
$\tau$ be a stopping time with values in
$(0,1]\cup\braces{\infty}$. We remind the reader that an adapted
process $U$ with values in $(-\infty,0)^M$ defined on $[0,\tau)\cap
[0,1]$ is called a \emph{maximal local solution} to~\eqref{eq:10} with
the \emph{explosion time} $\tau$ if for every stopping time $\sigma$
with values in $[0,\tau)\cap [0,1]$ the process $U$
satisfies~\eqref{eq:10} on $[0,\sigma]$ and if in addition
\begin{displaymath}
  \lim_{t\uparrow \tau} \max_{m=1,\dots,M}U^m_t = 0 \mtext{on} \{\tau <
  \infty\}.  
\end{displaymath}
Note that, since a negative local martingale is a submartingale,
$\lim_{t\uparrow \tau} U^m_t$ exists and is finite.

Recall the notation $a_m = a_m(x)$ from \eqref{eq:1} for the absolute
risk-aversion of $u_m = u_m(x)$.

\begin{Theorem}
  \label{th:1}
  Let Assumptions~\ref{as:1} and~\ref{as:2} and condition~\eqref{eq:6}
  hold.  Denote by $l$ the smallest integer such that
  \begin{displaymath}
    l > \frac{M+J}{2},
  \end{displaymath}
  and suppose $u_m\in \mathbf{C}^{l+2}$ with
  \begin{equation}
    \label{eq:11}
    \sup_{x\in
      \mathbf{R}}\abs{a^{(k)}_m(x)} < \infty, \quad k=1,\dots,l, \quad
    m=1,\dots,M.   
  \end{equation}
  Then for every locally bounded predictable process $Q$ with values
  in $\mathbf{R}^J$ there is a unique maximal local solution
  to~\eqref{eq:10}.
\end{Theorem}

Lemma~\ref{lem:7} contains equivalent reformulations
of~\eqref{eq:11}. The proof of Theorem~\ref{th:1} as well as of
Theorems~\ref{th:2} and~\ref{th:3} below will be given in
Section~\ref{sec:proofs}.

Clearly, the conditions of Theorem~\ref{th:1} are satisfied for the
exponential utilities:
\begin{equation}
  \label{eq:12}
  u_m(x) = -\frac1{a_m} e^{-a_m x}, \quad x\in \mathbf{R}, \quad
  m=1,\dots, M,  
\end{equation}
where $a_m$ is a positive number.  Direct computations show that in
this case
\begin{equation}
  \label{eq:13}
  r(v,x) = -\frac1{a} e^{-ax} \prod_{m=1}^M \left.v_m\right.^{\frac{a}{a_m}}, \quad
  (v,x)\in (0,\infty)^M\times \mathbf{R},  
\end{equation}
where the constant $a>0$ is given by
\begin{equation}
  \label{eq:14}
  \frac1a = \sum_{m=1}^M \frac1{a_m}.
\end{equation}
The integrability condition~\eqref{eq:6} takes now the form:
\begin{equation}
  \label{eq:15}
  \mathbb{E}[e^{-a\Sigma_0 + p\abs{\psi}}] < \infty, \quad p>0. 
\end{equation}
In fact, for exponential utilities we have a stronger (global) result.

\begin{Theorem}
  \label{th:2}
  Let Assumption~\ref{as:2} and conditions~\eqref{eq:12}
  and~\eqref{eq:15} hold. Then for every locally bounded and
  predictable process $Q$ with values in $\mathbf{R}^J$ there is a
  unique solution to~\eqref{eq:10}.
\end{Theorem}

Our final result provides conditions for the existence and uniqueness
of a solution of~\eqref{eq:10} in terms of the Malliavin derivatives
of $\Sigma_0$ and $\psi = (\psi^m)$.  We refer the reader to
\cite{Nualart:06} for an introduction to the Malliavin Calculus and
the notation used in the sequel. For $p\geq 1$ we denote by
$\mathbf{D}^{1,p}$ the Banach space of random variables $\xi$ with
Malliavin derivative $D\xi = (D_t\xi)_{t\in [0,1]}$ and the norm:
\begin{displaymath}
  \norm{\xi}_{\mathbf{D}^{1,p}} \set \left(\mathbb{E}[\abs{\xi}^p] +
    \mathbb{E}\left[\left(\int_0^1 \abs{D_t\xi}^2
        dt\right)^{p/2}\right]\right)^{1/p}.   
\end{displaymath}

\begin{Theorem}
  \label{th:3}
  In addition to Assumption~\ref{as:1} suppose $u_m\in \mathbf{C}^3$
  and
  \begin{equation}
    \label{eq:16} 
    \sup_{x\in \mathbf{R}}\abs{a'_m(x)} < \infty, \quad m=1,\dots,M. 
  \end{equation}
  Assume also that the filtration is generated by a $d$-dimensional
  Brownian motion $B=(B^i)$ and that the random variables $\Sigma_0$
  and $\psi=(\psi^j)$ belong to $\mathbf{D}^{1,2}$ and satisfy the
  integrability condition
  \begin{equation}
    \label{eq:17}
    \mathbb{E}[e^{ p\abs{\psi} + 2c\Sigma_0^{-}/M}] < \infty,  \quad \quad p>0, 
  \end{equation}
  with the constant $c>0$ from~\eqref{eq:1}.  Then for every locally
  bounded and predictable process $Q$ with values in $\mathbf{R}^J$
  there is a unique solution to~\eqref{eq:10}.
\end{Theorem}

For a corollary to Theorem~\ref{th:3} we consider the case where
$\Sigma_0$ and $\psi$ are defined in terms of the solution $X$ to the
stochastic differential equation:
\begin{equation}
  \label{eq:18}
  X_t = X_0 + \int_0^t \mu(s,X_s)ds + \int_0^t \sigma(s,X_s)dB_s,
  \quad X_0 \in \mathbf{R}^N.    
\end{equation}
We assume that the functions $\map{\mu}{[0,1]\times
  \mathbf{R}^N}{\mathbf{R}^N}$ and
$\map{\sigma}{[0,1]\times\mathbf{R}^N}{\mathbf{R}^{N \times d}}$ are
Lipschitz-continuous with respect to $x$ and bounded, i.e., there is a
constant $k>0$ such that for all $x,y\in \mathbf{R}^N$ and $t\in
[0,1]$
\begin{equation}
  \label{eq:19}
  \begin{split}
    \abs{\sigma(t,x) - \sigma(t,y)} + \abs{\mu(t,x) - \mu(t,y)} &\leq
    k\abs{x-y}, \\
    \abs{\sigma(t,x)} + \abs{\mu(t,x)} &\leq k.
  \end{split}
\end{equation}

\begin{Corollary}
  In addition to Assumption~\ref{as:1} suppose $u_m\in \mathbf{C}^3$
  and~\eqref{eq:16} holds.  Assume also that the filtration is
  generated by a $d$-dimensional Brownian motion $B=(B^i)$ and that
  the random variables $\Sigma_0$ and $\psi=(\psi^j)$ are of the form
  \begin{displaymath}
    \Sigma_0 = g(X_1), \quad \psi^j=f^j(X_1), \quad j=1,\dots,J, 
  \end{displaymath}
  for some Lipschitz-continuous functions $g$ and $f=(f^j)$ on
  $\mathbf{R}^N$ and for the solution $X$ to the stochastic
  differential equation~\eqref{eq:18} with coefficients
  satisfying~\eqref{eq:19}.  Then for every locally bounded
  predictable process $Q$ with values in $\mathbf{R}^J$ there is a
  unique solution to~\eqref{eq:10}.
\end{Corollary}
\begin{proof}
  It is well-known, see Theorem~2.2.1 in \cite{Nualart:06}, that under
  the stated assumptions $X_1 \in \mathbf{D}^{1,2}$. By the chain rule
  of Malliavin calculus this also holds for the Lipschitz-continuous
  transformations $\Sigma_0$ and $\psi$ of $X_1$; see
  Proposition~1.2.4 in \cite{Nualart:06}. As the coefficients $\mu$
  and $\sigma$ are bounded, $X_1$ and, therefore, $\Sigma_0$ and $\psi
  = (\psi^j)$ have finite exponential moments of any order. The
  assertion now follows from Theorem~\ref{th:3}.
\end{proof}

\section{Conditions in terms of SDE-coefficients}
\label{sec:conditions-K-H}

In this section we state solvability criteria for~\eqref{eq:10} in
terms of the stochastic fields $K$ and $H$. These conditions will be
used later in the proofs of the main theorems.

\begin{Lemma}
  \label{lem:1}
  Suppose the process $K$ defined in \eqref{eq:9} takes values in
  $\mathbf{C}^1((-\infty,0)^M \times \mathbf{R}^J, \mathbf{R}^{M\times
    d})$ and for every compact set $C\subset (-\infty,0)^M \times
  \mathbf{R}^J$
  \begin{displaymath}
    \int_0^1 \norm{K_t}^2_{1,C} dt  < \infty. 
  \end{displaymath}
  If $Q$ is a predictable process with values in $\mathbf{R}^J$ such
  that for every compact set $C\subset (-\infty,0)^M$
  \begin{displaymath}
    \int_0^1 \norm{K_t(\cdot,Q_t)}^2_{1,C} dt  < \infty,
  \end{displaymath}
  then there is a unique maximal local solution $U$
  to~\eqref{eq:10}. In particular, such a solution exists for every
  locally bounded predictable $Q$.
\end{Lemma}
\begin{proof}
  Follows from well-known criteria for maximal local solutions; see
  Theorem 3.4.5 in \cite{Kunit:90}.
\end{proof}

For vectors $x,y \in \mathbf{R}^M$ we shall write $x\geq y$ if
$x^m\geq y^m$, $m=1,\dots,M$. Denote $\idvec \set (1,\dots,1)$.

\begin{Lemma}
  \label{lem:2}
  Let the processes $K$ and $Q$ satisfy the conditions of
  Lemma~\ref{lem:1} and suppose that for every constant $b>0$
  \begin{equation}
    \label{eq:20}
    \int_0^1 \sup_{-b \idvec \leq u < 0} L_t(u,Q_t) dt < \infty,
  \end{equation}
  where, for $u\in (-\infty,0)^M$, $q\in \mathbf{R}^J$, and $t\in
  [0,1]$,
  \begin{equation}
    \label{eq:21}
    L_t(u,q) \set \frac1{1 + \sum_{m=1}^M \abs{\log(-u^m)}} 
    \sum_{m=1}^M \left|\frac{{K^m_t(u,q)}}{u^m}\right|^2. 
  \end{equation}
  Then~\eqref{eq:10} has a unique (global) solution.
\end{Lemma}

\begin{proof}
  In view of Lemma~\ref{lem:1} we have to show that the explosion time
  $\tau$ for the maximal local solution $U$ to~\eqref{eq:10} is
  infinite.  By localization and accounting for the submartingale
  property of $U$ we can assume without a loss in generality that
  $U^m\geq -b$ for some $b>0$.
 
  After the substitution $U^m = -\exp(Z^m)$, $m=1,\dots,M$, we can
  rewrite~\eqref{eq:10} as
  \begin{displaymath}
    Z^m_t = Z^m_0 + \int_0^t A^m_s(Z_s) dB_s - \frac12 \int_0^t
    \abs{A^m_s(Z_s)}^2 ds, 
  \end{displaymath}
  where
  \begin{displaymath}
    A^m_t(z) \set - e^{-z^m} K^m_t(-e^{z^1}, \dots,-e^{z^M},Q_t),
    \quad m=1,\dots,M. 
  \end{displaymath}
  Each component of $Z$ is bounded from above by $\log b$ and
  \begin{displaymath}
    \lim_{t\uparrow \tau} \min_{m=1,\dots,M}Z^m_t = -\infty \mtext{on} \{\tau <
    \infty\}.  
  \end{displaymath}
  It is well-known, see \cite[Theorem 3.4.6]{Kunit:90}, that
  $\tau=\infty$ if
  \begin{displaymath}
    \int_0^{1} \sup_{z\leq \idvec\log b}\frac{\abs{A^m_t(z)}^2}{1 +
      \abs{z}} dt < \infty, \quad m=1, \dots, M, 
  \end{displaymath}
  which is equivalent to~\eqref{eq:20}.
\end{proof}

When $Q$ is locally bounded, we can state more convenient conditions
in terms of the ``primal'' processes $F$ and $H$. Recall first a
result from \cite{BankKram:11b}; see Lemma~5.15 and Theorems~5.16
and~5.17.

\begin{Lemma}[\cite{BankKram:11b}]
  \label{lem:3}
  Let Assumptions~\ref{as:1} and~\ref{as:2} and condition~\eqref{eq:6}
  hold. Suppose the process $H$ of~\eqref{eq:8} has values in
  $\mathbf{C}^2(\mathbf{A}, \mathbf{R}^d)$ and for every compact set
  $C\subset \mathbf{A}$
  \begin{equation}
    \label{eq:22}
    \int_0^1 \norm{H_t}^2_{2,C} dt < \infty. 
  \end{equation}
  Then the process $K$ of~\eqref{eq:9} satisfies the conditions of
  Lemma~\ref{lem:1} and for every locally bounded predictable $Q$ with
  values in $\mathbf{R}^J$ there is a unique maximal local solution
  to~\eqref{eq:10}.
\end{Lemma}

\begin{Remark}
  \label{rem:1}
  As $H(bv,x,q) = bH(v,x,q)$ for $b>0$ and $(v,x,q)\in \mathbf{A}$, it
  is sufficient to verify~\eqref{eq:22} for compact sets $C \subset
  \widetilde{\mathbf{A}} \set \mathbf{S}^M \times \mathbf{R} \times
  \mathbf{R}^J$, where $\mathbf{S}^M$ is the interior of the simplex
  in $\mathbf{R}^M$:
  \begin{displaymath}
    \mathbf{S}^M \set \descr{w \in (0,1)^M}{\sum_{m=1}^M w^m = 1}.
  \end{displaymath}
\end{Remark}

We now state criteria for the existence of a global solution. In its
proof we shall make use of the conjugacy relations between the
stochastic fields $F$ and $G$, which for $a= (v,x,q) \in \mathbf{A}$,
$b=(u,1,q)\in \mathbf{B}$, and $t\in [0,1]$ state that
\begin{align}
  \label{eq:23}
  v &= \frac{\partial G_t}{\partial u}\left(\frac{\partial
      F_t}{\partial
      v}(a),1,q\right), \\
  \label{eq:24}
  x &= G_t\left(\frac{\partial F_t}{\partial v}(a),1,q\right), \\
  \label{eq:25}
  u &= \frac{\partial F_t}{\partial v}\left(\frac{\partial
      G_t}{\partial u}(b),G_t(b),q\right);
\end{align}
see Corollary 4.14 in \cite{BankKram:11b}.

\begin{Lemma}
  \label{lem:4}
  Assume the conditions of Lemma~\ref{lem:3}. Suppose also that
  \begin{equation}
    \label{eq:26}
    \int_0^1 \sup_{a\in A_t(b)} M_t(a) dt < \infty,\quad b>0,
  \end{equation}
  where, for $t\in [0,1]$ and a constant $b>0$,
  \begin{displaymath}
    A_t(b) \set \descr{a=(v,x,q)\in \mathbf{A}}{\frac{\partial
        F_t}{\partial v}(a)\geq -b\idvec,\; 
      \abs{q}\leq b},  
  \end{displaymath}
  and, for $a=(v,x,q)\in \mathbf{A}$,
  \begin{displaymath}
    M_t(a) \set  \frac1{1+\abs{x}}
    \sum_{m=1}^M \left|\frac{1}{\frac{\partial F_t}{\partial v^m}(a)} 
      \frac{\partial H_t}{\partial v^m}(a)\right|^2.  
  \end{displaymath}
  Then for every locally bounded predictable $Q$ with values in $
  \mathbf{R}^J$ there is a unique (global) solution to~\eqref{eq:10}.
\end{Lemma}

\begin{Remark}
  \label{rem:2}
  Since for all $(v,x)\in (0,\infty)^M\times \mathbf{R}$ and
  $m=1,\dots,M$
  \begin{displaymath}
    \frac1c \frac{\partial r}{\partial x}(v,x)  \leq -v^m
    \frac{\partial r}{\partial 
      v^m}(v,x)\leq c \frac{\partial r}{\partial x}(v,x),   
  \end{displaymath}
  see Section~4.1 in \cite{BankKram:11a}, the bound~\eqref{eq:26} is
  equivalent to
  \begin{displaymath}
    \int_0^1 \sup_{a\in \widetilde A_t(b)} \widetilde M_t(a) dt <
    \infty, \quad b>0,
  \end{displaymath}
  where
  \begin{align*}
    \widetilde A_t(b) &\set \descr{a=(v,x,q)\in
      \mathbf{A}}{\frac{\partial
        F_t}{\partial x}(a)\idvec\leq bv,\; \abs{q}\leq b},   \\
    \widetilde M_t(a) &\set \frac1{(1+\abs{x})(\frac{\partial
        F_t}{\partial x}(a))^2} \sum_{m=1}^M \left|v^m \frac{\partial
        H_t}{\partial v^m}(a)\right|^2.
  \end{align*}
\end{Remark}

\begin{proof}
  Let $c>0$ denote the constant appearing in~\eqref{eq:1}. By
  Lemma~\ref{lem:2} it is enough to show that for $L = L_t(u,q)$
  defined in~\eqref{eq:21} and every $b\geq c$
  \begin{displaymath}
    \int_0^1 \sup_{-b\idvec \leq u < 0, \; \abs{q}\leq b} L_t(u,q) dt < \infty.  
  \end{displaymath}
  
  By Corollary~5.4 in \cite{BankKram:11b} the process $G=G_t(b)$ has
  trajectories in $\mathbf{C}([0,1],\widetilde{\mathbf{G}}^2(c))$,
  where $\widetilde{\mathbf{G}}^2(c)$ is a linear subspace of saddle
  functions in $\mathbf{C}^2(\mathbf{B})$ defined and studied in
  Section~3 of \cite{BankKram:11a}. Property (G7) of the elements of
  $\widetilde{\mathbf{G}}^2(c)$ yields that
  \begin{displaymath}
    \frac1b \leq \frac1c \leq -u^m \frac{\partial G_t}{\partial
      u^m}(u,1,q) \leq c \leq b, \quad m=1,\dots, M,  
  \end{displaymath}
  which for $-b\idvec \leq u < 0$ implies
  \begin{equation}
    \label{eq:27}
    b  \sum_{m=1}^M \log(-u^m/b) \leq G_t(-b\idvec,1,q) - G_t(u,1,q)
    \leq \frac1b \sum_{m=1}^M \log(-u^m/b). 
  \end{equation}
  
  For $n\geq 1$ define the stopping times
  \begin{displaymath}
    \sigma_n \set \inf\descr{t\in [0,1]}{\sup_{\abs{q}\leq
        b}\abs{G_t(-b\idvec,1,q)} > n},  
  \end{displaymath}
  where, by convention, $\inf \varnothing\set\infty$.  Since the
  sample paths of $G =G_t(b)$ belong to
  $\mathbf{C}([0,1],\mathbf{C}(\mathbf{B}))$, we deduce $\sigma_n \to
  \infty$, $n\to \infty$. Hence, the result holds if
  \begin{displaymath}
    \int_0^{\sigma_n\wedge 1} \sup_{-b\idvec \leq u < 0, \;
      \abs{q}\leq b} L_t(u,q) dt < \infty, \quad n\geq 1.  
  \end{displaymath}
  Observe now that the conjugacy
  relations~\eqref{eq:23}--\eqref{eq:25} jointly with~\eqref{eq:27}
  and the construction of $K = K_t(u,q)$ in~\eqref{eq:9} imply that
  for $t \leq \sigma_n \wedge 1$
  \begin{displaymath}
    \sup_{-b\idvec < u < 0, \; \abs{q}\leq b} L_t(u,q) \leq c(b,n)
    \sup_{a=(v,x,q) \in A_t(b)} M_t(a),
  \end{displaymath}
  where the constant $c(b,n)$ depends only on $b$ and $n$. The result
  now follows from~\eqref{eq:26}.
\end{proof}

The conditions of Lemma~\ref{lem:4} can be simplified further if
instead of~\eqref{eq:6} we assume the stronger integrability
condition~\eqref{eq:7}.

\begin{Lemma}
  \label{lem:5}
  Let Assumptions~\ref{as:1} and~\ref{as:2} and
  conditions~\eqref{eq:7} and~\eqref{eq:22} hold. Then \eqref{eq:26}
  is equivalent to
  \begin{equation}
    \label{eq:28}
    \int_0^1 \sup_{a\in B(b)} N_t(a) dt < \infty, \quad b>0, 
  \end{equation}
  where, for a constant $b>0$,
  \begin{displaymath}
    B(b) \set \descr{a=(v,x,q)\in  \mathbf{A}}{
      \frac{\partial r}{\partial x}(v,x) \idvec 
      \leq bv,\; \abs{q}\leq b},  
  \end{displaymath}
  and, for $a=(v,x,q)\in \mathbf{A}$,
  \begin{displaymath}
    N_t(a) \set  \frac1{(1+\abs{x})(\frac{\partial r}{\partial
        x}(v,x))^2} \sum_{m=1}^M \left|v^m \frac{\partial
        H_t}{\partial v^m}(a)\right|^2.   
  \end{displaymath}
  In this case, for every locally bounded predictable $Q$ with values
  in $ \mathbf{R}^J$ there is a unique solution to~\eqref{eq:10}.
\end{Lemma}

\begin{proof}
  The result follows from Lemma~\ref{lem:4} and Remark~\ref{rem:2} if
  for every $b>0$ we can find strictly positive random variables
  $\eta$ and $\zeta$ such that
  \begin{displaymath}
    \zeta \frac{\partial r}{\partial x}(v,x) \leq 
    \frac{\partial F_t}{\partial x}(v,x,q) \leq 
    \eta \frac{\partial r}{\partial x}(v,x), 
  \end{displaymath}
  for all $(v,x)\in (0,\infty)^M \times \mathbf{R}$, $\abs{q}\leq b$,
  and $t\in [0,1]$.

  From~\eqref{eq:4} we deduce for such $(v,x,q)$ that
  \begin{displaymath}
    e^{-c(\Sigma_0^+ + b\abs{\psi})/M} \frac{\partial
      r}{\partial x}(v,x) \leq \frac{\partial r}{\partial x}(v,\Sigma(x,q))
    \leq e^{c (\Sigma_0^{-} + b\abs{\psi})/M} \frac{\partial
      r}{\partial x}(v,x).
  \end{displaymath}
  The random variables $\zeta$ and $\eta$ can now be chosen as
  \begin{align*}
    \zeta &\set \inf_{t\in [0,1]} \mathbb{E}[e^{-c( \Sigma_0^+ +
      b \abs{\psi})/M} |\mathcal{F}_t], \\
    \eta &\set \sup_{t\in [0,1]} \mathbb{E}[ e^{c(\Sigma_0^- + b
      \abs{\psi})/M}|\mathcal{F}_t].
  \end{align*}
  Note that $\zeta$ is strictly positive and, in view of~\eqref{eq:7},
  $\eta$ is finite.

\end{proof}
 
\section{Proofs of the main results}
\label{sec:proofs}

\subsection{Proof of Theorem~\ref{th:1}}
\label{sec:proof-theor-refth:1}

We divide the proof into a series of lemmas. A key role is played by
the following direct corollary of Proposition~1 in \cite{Szn:81}.

\begin{Lemma}
  \label{lem:6} Suppose Assumption~\ref{as:2} holds. Let $V$ be an
  open set in $\mathbf{R}^n$ and $m$, $j$ be nonnegative integers
  satisfying
  \begin{displaymath} m-\frac{n}2> j.
  \end{displaymath} 
  Consider a random field $\eta = (\eta(x))_{x\in V}$ with sample
  paths in $\mathbf{C}^m(V)$ such that for every compact set $C\subset
  V$
  \begin{displaymath} \mathbb{E}[\norm{\eta}_{m,C}] < \infty.
  \end{displaymath} 

  Then there exists a predictable process $H$ with values in
  $\mathbf{C}^j(V,\mathbf{R}^{d})$ such that for every $x\in V$ and
  $t\in [0,1]$ and every multi-index $\beta=(\beta_1,\dots,\beta_n)$
  of order $0 \leq \abs{\beta}\leq j$
  \begin{displaymath}
    \mathbb{E}[\partial^\beta \eta(x)|\mathcal{F}_t] =
    \mathbb{E}[\partial^\beta \eta(x)] + \int_0^t \partial^\beta
    H_u(x)dB_u. 
  \end{displaymath}
  Moreover, for every compact set $C\subset V$,
  \begin{displaymath}
    \int_0^1 \norm{H_t}^2_{j,C} dt < \infty.
  \end{displaymath}
\end{Lemma}

\begin{proof} Without restricting generality we can assume that $V$ is
  a ball in $\mathbf{R}^d$ and
  \begin{displaymath}
    \mathbb{E}[\norm{\eta}_{m,V}] < \infty.
  \end{displaymath}
  Assumption~\ref{as:2} implies that
  \begin{displaymath} L_t \set
    \mathbb{E}[\norm{\eta}_{m,V}|\mathcal{F}_t], \quad t\in [0,1],
  \end{displaymath} is a continuous martingale. Hence, there is a
  sequence of stopping times $(\tau_n)_{n\geq 1}$ such that
  $\braces{\tau_n<1}\downarrow \emptyset$ and $\abs{L}\leq n$ 
  on $[0,\tau_n]$.  By Lemma~4.5 in \cite{BankKram:11b}, the random field
  \begin{displaymath} \eta_n(x) \set
    \mathbb{E}[\eta(x)|\mathcal{F}_{\tau_n}], \quad x\in V,
  \end{displaymath} has a version with values in $\mathbf{C}^m$ and
  $\norm{\eta_n}_{m,V} \leq L_{\tau_n} \leq n$.

  These localization arguments imply that without any loss of
  generality we can assume that the random variable $
  \norm{\eta}_{m,V}$ is bounded and, in particular,
  \begin{displaymath}
    \mathbb{E}[\norm{\eta}^2_{m,V}] < \infty.
  \end{displaymath} 
  However, under this condition the assertion of the lemma is a
  special case of Proposition 1 in \cite{Szn:81}.
\end{proof}

\begin{Lemma}
  \label{lem:7}
  Let Assumption~\ref{as:1} hold and suppose each $u_m$,
  $m=1,\ldots,M$, is of class $\mathbf{C}^{l+2}$ for an integer $l\geq
  0$. Then~\eqref{eq:11} is equivalent to the condition
  \begin{displaymath}
    \sup_{x\in \mathbf{R}}
    \left.\frac{\abs{u^{(k)}_m(x)}}{u'_m(x)}\right. <
    \infty, \quad k=0,\dots,l+2,\quad m=1,\dots,M,
  \end{displaymath}
  and also to the condition
  \begin{equation}
    \label{eq:29}
    \sup_{x\in \mathbf{R}} \abs{t^{(k)}_m(x)} < \infty, \quad
    k=0,\dots,l,\quad m=1,\dots,M,
  \end{equation}
  where $t_m(x) \set 1/a_m(x)$ is the absolute risk-tolerance of $u_m
  = u_m(x)$.
\end{Lemma}
\begin{proof}
  Follows from Assumption~\ref{as:1} by direct computations.
\end{proof}

\begin{Lemma}
  \label{lem:8}
  Let Assumption~\ref{as:1} hold and suppose that each $u_m$,
  $m=1,\ldots,M$, is of class $\mathbf{C}^{l+2}$ and
  that~\eqref{eq:11} holds for an integer $l\geq 0$. Then the function
  $r=r(v,x)$ is of class $\mathbf{C}^{l+2}$ and there is a constant
  $b>0$ such that for every multi-index $\beta=(\beta_1,\dots,
  \beta_M,\beta_{M+1})$ of non-negative integers with $\abs{\beta}\leq
  l+2$
  \begin{equation}
    \label{eq:30}
    \abs{\mathbb{T}^\beta r(v,x)} \leq b \frac{\partial
      r}{\partial x}(v,x), \quad (v,x)\in (0,\infty)^M\times\mathbf{R},
  \end{equation}
  where $\mathbb{T}^{\beta}$ is the differential operator
  \begin{displaymath}
    \mathbb{T}^{\beta} \set 
    \left(\prod_{m=1}^{M}(v^m \frac{\partial}{\partial v^m})^{\beta_m} \right)
    \frac{\partial^{\beta_{M+1}}}{\partial x^{\beta_{M+1}}}.  
  \end{displaymath}
\end{Lemma}
\begin{proof}
  If $\abs{\beta}\leq 2$, then the inequality \eqref{eq:30} follows
  from Theorems~4.1 and~4.2 in \cite{BankKram:11a} containing explicit
  expressions for $\mathbb{T}^{\beta} r(v,x)$ in this case. If $2 <
  \abs{\beta} \leq l+2$, then elementary computations based on the
  formulas from Theorem~4.2 in \cite{BankKram:11a} yield
  \begin{displaymath}
    \mathbb{T}^{\beta}r(v,x) = \frac{\partial r}{\partial x}(v,x)
    P_{\beta}(v,x) \frac1{(\sum_{m=1}^M t_m(\widehat x^m))^{2\abs{\beta}}}, 
  \end{displaymath}
  where $(\widehat{x}^m)_{m=1,\dots,M}$ is the maximizer
  in~\eqref{eq:2} and $P_{\beta} = P_{\beta}(v,x)$ is a polynomial of
  $t^{(n)}_m(\widehat x^m)$, $n=0,\dots,\abs{\beta}-2$, $m=1,\dots,
  M$. The inequality~\eqref{eq:30} now follows from \eqref{eq:29} and
  Assumption~\ref{as:1}.
\end{proof}

\begin{Lemma}
  \label{lem:9}
  Assume the conditions of Lemma~\ref{lem:8} and denote
  \begin{equation}
    \label{eq:31}
    \xi(a) \set r(v,\Sigma(x,q)), \quad a= (v,x,q)\in \mathbf{A}. 
  \end{equation}
  Then for every compact set $C\subset \mathbf{A}$ there are a
  constant $b>0$ and a compact set $D\subset \mathbf{A}$ containing
  $C$ such that
  \begin{displaymath}
    \norm{\xi}_{l+2,C} \leq b \norm{\xi}_{D} \set b \sup_{a\in D}
    \abs{\xi(a)}.  
  \end{displaymath}
\end{Lemma}
\begin{proof}
  Lemma~\ref{lem:8} implies the existence of a constant $b>0$ such
  that
  \begin{equation}
    \label{eq:32}
    \norm{\xi}_{l+2,C} \leq b \sup_{(v,x,q)\in C} \frac{\partial
      r}{\partial x}(v,\Sigma(x,q)) (1 + \abs{\psi})^{l+2}. 
  \end{equation}
  In view of~\eqref{eq:4} the right side of~\eqref{eq:32} is dominated
  by $b_1\norm{\xi}_{1,E}$ for some constant $b_1>0$ and a compact set
  $E$ in $\mathbf{A}$ containing $C$. Since the sample paths of $\xi =
  \xi(a)$ are saddle functions, $\norm{\xi}_{1,E}$ is dominated by
  $b_2\norm{\xi}_{D}$, where $D\subset \mathbf{A}$ is a compact set
  whose interior contains $E$ and $b_2 = b_2(E,D)$ is a positive
  constant depending only on $E$ and $D$.
\end{proof}

The proof of Theorem~\ref{th:1} now follows from Lemma~\ref{lem:3} and

\begin{Lemma}
  \label{lem:10}
  Under the assumptions of Theorem~\ref{th:1} the process $H$
  of~\eqref{eq:8} has values in $\mathbf{C}^2(\mathbf{A},
  \mathbf{R}^d)$ and~\eqref{eq:22} holds for every compact set
  $C\subset \mathbf{A}$.
\end{Lemma}
\begin{proof}
  Recall the notation $\widetilde{\mathbf{A}}$ from Remark~\ref{rem:1}
  and observe that the dimension of this set is $M+J$.  In view of
  Lemma~\ref{lem:6} and Remark~\ref{rem:1} it is enough to show that
  for every compact set $C\subset \widetilde{\mathbf{A}}$ and an
  integer $l>(M+J)/2$ the random field $\xi$ of~\eqref{eq:31}
  satisfies
  \begin{displaymath}
    \mathbb{E}[\norm{\xi}_{l+2,C}] < \infty.
  \end{displaymath}
  This follows from Lemma~\ref{lem:9} and the fact that~\eqref{eq:6}
  implies the integrability of $\norm{\xi}_{D}$ for every compact set
  $D\subset \mathbf{A}$; see Lemma~4.12 in \cite{BankKram:11a}.
\end{proof}

\subsection{Proof of Theorem~\ref{th:2}}
\label{sec:proof-theor-refth:2}

It is enough to verify the growth condition~\eqref{eq:26} of
Lemma~\ref{lem:4}.  For $q\in \mathbf{R}^J$ define the processes
$\widetilde F(q)$ and $\widetilde H(q)$ by
\begin{displaymath}
  \widetilde F_t(q) \set \mathbb{E}[e^{-a(\Sigma_0 +
    \ip{q}{\psi})}|\mathcal{F}_t] = \widetilde F_0(q) + \int_0^t \widetilde
  H_s(q) dB_s,
\end{displaymath}
with the constant $a>0$ from~\eqref{eq:14}, and observe that,
by~\eqref{eq:13},
\begin{align*}
  F_t(v,x,q) & = r(v,x) \widetilde F_t(q), \\
  H_t(v,x,q) & = r(v,x) \widetilde H_t(q).
\end{align*}
It follows that \eqref{eq:26} holds if for every $b>0$
\begin{displaymath}
  \int_0^1 \sup_{\abs{q}\leq b} \left(\frac{\abs{\widetilde H_t(q)}}{\widetilde
      F_t(q)}\right)^2 dt < \infty.
\end{displaymath}
This inequality holds since
\begin{displaymath}
  \inf_{t\in [0,1]}\inf_{\abs{q}\leq b} \widetilde F_t(q) \geq  \inf_{t\in [0,1]}
  \mathbb{E}[\inf_{\abs{q}\leq b}e^{-a(\Sigma_0 + \ip{q}{\psi})}|\mathcal{F}_t] =
  \inf_{t\in [0,1]}\mathbb{E}[e^{-a(\Sigma_0 + 
    b\abs{\psi})}|\mathcal{F}_t] >0,
\end{displaymath}
and because, in view of Lemma~\ref{lem:10},
\begin{displaymath}
  \int_0^1 \sup_{\abs{q}\leq b} \abs{\widetilde H_t(q)}^2 dt  <\infty.  
\end{displaymath}

This ends the proof of Theorem~\ref{th:2}.

\subsection{Proof of Theorem~\ref{th:3}}
\label{sec:proof-theor-refth:3}

The proof is divided into a series of lemmas, where we shall verify
the assumptions of Lemma~\ref{lem:5}.  Hereafter we denote
\begin{displaymath}
  \xi(a) \set r(v,\Sigma(x,q)), \quad a = (v,x,q) \in \mathbf{A}.
\end{displaymath}
As usual $\mathbf{L}^p$ stands for the space of $p$-integrable random
variables, $p\geq 1$.

\begin{Lemma}
  \label{lem:11}
  Suppose Assumption~\ref{as:1} and conditions~\eqref{eq:16}
  and~\eqref{eq:17} hold. Then $\norm{\xi}_{3,C} \in \mathbf{L}^2$ for
  every compact set $C\subset \mathbf{A}$.
\end{Lemma}
\begin{proof}
  Lemma~\ref{lem:9} and our boundedness assumptions on $a_m$ and
  $a'_m$ yield that $ \norm{\xi}_{3,C} \leq b_1 \norm{\xi}_{E}$
  for some constant $b_1>0$ and a compact set $E\subset \mathbf{A}$
  containing $C$. From~\eqref{eq:4} and~\eqref{eq:5} we deduce the
  existence of a constant $b_2>0$ such that
  \begin{displaymath}
    \norm{\xi}_{E} \leq b_2 e^{ b_2\abs{\psi} + c\Sigma_0^{-}/M}
  \end{displaymath}
  and the result follows from~\eqref{eq:17}.
\end{proof}

\begin{Lemma}
  \label{lem:12}
  Under the conditions of Theorem~\ref{th:3} we have $\xi(a) \in
  \mathbf{D}^{1,1}$, $a \in \mathbf{A}$, with the Malliavin derivative
  \begin{displaymath}
    D \xi(a) = \frac{\partial r}{\partial
      x}(v,\Sigma(x,q))(D \Sigma_0 + \ip{q}{D \psi})
  \end{displaymath}
  and for every compact set $C\subset \mathbf{A}$
  \begin{equation}
    \label{eq:33}
    \mathbb{E}\left[\left( \int_0^1 \norm{D_t \xi}^2_{2,C} dt
      \right)^{1/2}\right] < \infty.
  \end{equation}
\end{Lemma}
\begin{proof}
  Let $(f_n)_{n\geq 1}$ be a sequence of continuously differentiable
  functions on $(-\infty,0)$ such that $f_n(x) = x$ on $(-n,0)$ and
  $0\leq f'_n(x)\leq -n/x$ on $(-\infty,-n]$. For example, we can take
  \begin{displaymath}
    f_n(x) = x \ind{-n< x<0} - (2n+\frac{n^2}{x}) \ind{x \leq -n}. 
  \end{displaymath}
  The function $f_n(r(v,.))$ is continuously differentiable and, in
  view of~\eqref{eq:5}, has bounded derivatives.  By the chain rule of
  Malliavin calculus, see Proposition~1.2.3 in \cite{Nualart:06},
  $\xi_n(a) \set f_n(\xi(a))$ belongs to $\mathbf{D}^{1,2}$ with
  the Malliavin derivative $D \xi_n(a) = f'_n(\xi(a)) Z(a)$, where
  \begin{displaymath}
    Z_t(a) \set \frac{\partial r}{\partial
      x}(v,\Sigma(x,q))(D_t \Sigma_0 + \ip{q}{D_t \psi}), \quad t\in
    [0,1].
  \end{displaymath}
  
  By construction, $\xi(a)\leq \xi_n(a)<0$ and $\xi_n(a) \to \xi(a)$
  almost surely. This readily implies the convergence $\xi_n(a) \to
  \xi(a)$ in $\mathbf{L}^1$. Since $0\leq f'_n(x)\leq 1$ and
  $f_n'(x)\to 1$, the convergence $\xi_n(a) \to \xi(a)$ in
  $\mathbf{D}^{1,1}$ and the identity $D\xi(a) = Z(a)$ follow if we
  can show that for every compact set $C\subset \mathbf{A}$
  \begin{displaymath}
    \mathbb{E}\left[\left( \int_0^1 \norm{Z_t}^2_{2,C} dt
      \right)^{1/2}\right] < \infty.
  \end{displaymath}
  Note that this will also establish~\eqref{eq:33}.

  We have
  \begin{displaymath}
    \int_0^1  \norm{Z_t}^2_{2,C} dt \leq \norm{\xi}_{3,C}^2 \int_0^1
    (\abs{D_t \Sigma_0} + b \abs{D_t \psi} )^2 dt, 
  \end{displaymath}
  where $b \set \sup_{(v,x,q)\in C} \abs{q}$, and then, by the
  Cauchy-Schwarz inequality,
  \begin{align*}
    \mathbb{E}\left[\left( \int_0^1 \norm{Z_t}^2_{2,C} dt
      \right)^{1/2}\right] & \leq \mathbb{E}\left[\norm{\xi}_{3,C}
      \left( \int_0^1 (\abs{D_t \Sigma_0} + b \abs{D_t \psi} )^2 dt
      \right)^{1/2}\right] \\
    & \leq
    \left(\mathbb{E}\left[\norm{\xi}_{3,C}^2\right]\right)^{1/2}
    \left(\mathbb{E}\left[\int_0^1 (\abs{D_t \Sigma_0} + b \abs{D_t
          \psi} )^2 dt\right]\right)^{1/2},
  \end{align*}
  which is finite because of Lemma~\ref{lem:11} and because
  $\Sigma_0,\psi\in \mathbf{D}^{1,2}$.
\end{proof}
 
\begin{Lemma}
  \label{lem:13}
  Under the conditions of Theorem~\ref{th:3} the process $H=H_t(a)$
  of~\eqref{eq:8} has values in $\mathbf{C}^2(\mathbf{A},
  \mathbf{R}^d)$ and~\eqref{eq:22} holds for every compact set
  $C\subset \mathbf{A}$. Moreover, for every multi-index $\beta =
  (\beta^1, \dots, \beta^{J+M+1})$ with $\abs{\beta}\leq 2$
  \begin{equation}
    \label{eq:34}
    \partial^\beta H_t(a) = \mathbb{E}\left[\left.\partial^\beta \left(
          \frac{\partial r}{\partial 
            x}(v,\Sigma(x,q))(D_t \Sigma_0 + \ip{q}{D_t
            \psi})\right)\right|\mathcal{F}_t\right], \; t\in [0,1]. 
  \end{equation}
\end{Lemma}
\begin{proof}
  Fix a compact set $C\subset \mathbf{A}$. Since, by
  Lemma~\ref{lem:12}, $\xi(a) \in \mathbf{D}^{1,1}$ the Clark-Ocone
  formula from \cite{KaratOconeLi:91} yields
  \begin{displaymath}
    H_t(a) = {\mathbb{E}}[ D_t \xi(a) |\mathcal{F}_t],
    \quad a\in \mathbf{A}, \; t\in [0,1], 
  \end{displaymath}
  or, equivalently,
  \begin{displaymath}
    H(a) \set \widetilde{\mathbb{E}}[ D \xi(a) |\mathcal{P}],
    \quad a\in \mathbf{A},
  \end{displaymath}
  where $\widetilde{\mathbb{P}} \set \mathbb{P} \otimes dt$ and
  $\mathcal{P}$ is the $\sigma$-algebra of predictable sets. From
  Lemma~\ref{lem:12} we deduce
  \begin{displaymath}
    \widetilde{\mathbb{E}}[\norm{D \xi}_{2,C}] \set
    \mathbb{E}[\int_0^1 \norm{D_t \xi}_{2,C} dt] \leq
    \mathbb{E}\left[\left( \int_0^1 \norm{D_t \xi}^2_{2,C} dt 
      \right)^{1/2}\right] < \infty.
  \end{displaymath}
  As a result, see Lemma~4.5 of \cite{BankKram:11b}, the random field
  $H=H(a)$ can be chosen with sample paths in
  $\mathbf{C}^2(\mathbf{A},\mathbf{R}^d)$ and, for a multi-index
  $\beta = (\beta^1, \dots, \beta^{J+M+1})$ with $\abs{\beta}\leq 2$,
  \begin{displaymath}
    \partial^\beta H(a) = \widetilde{\mathbb{E}}[ \partial^\beta D
    \xi(a) |\mathcal{P}], \quad a\in \mathbf{A},
  \end{displaymath}
  which is just a reformulation of~\eqref{eq:34}.

  It only remains to verify~\eqref{eq:22}. From~\eqref{eq:34} we
  obtain
  \begin{align*}
    \norm{H_t}^2_{2,C} &\leq
    b\left({\mathbb{E}}\left[\norm{\xi}_{3,C}(\abs{D_t\Sigma_0} +
        \abs{D_t\psi})|\mathcal{F}_t\right]\right)^2 \\
    & \leq b{\mathbb{E}}\left[\norm{\xi}_{3,C}^2 |\mathcal{F}_t\right]
    \mathbb{E}\left[(\abs{D_t\Sigma_0} +
      \abs{D_t\psi})^2|\mathcal{F}_t\right],
  \end{align*}
  for some constant $b$ depending on $C$.  The result now follows
  because, in view of Lemma~\ref{lem:11},
  $\mathbb{E}[\norm{\xi}_{3,C}^2|\mathcal{F}_t]$, $t \in [0,1]$, is a
  martingale and thus has bounded paths and because
  \begin{equation}
    \label{eq:35}
    \mathbb{E}[\int_0^1 (\abs{D_t \Sigma_0}^2 +
    \abs{D_t\psi}^{2}) dt] < \infty
  \end{equation}
  as $\Sigma_0,\psi\in \mathbf{D}^{1,2}$.
\end{proof}
  
To conclude the proof of Theorem~\ref{th:3} it only remains to verify
the growth condition~\eqref{eq:28} of Lemma~\ref{lem:5}. This is
accomplished in

\begin{Lemma}
  \label{lem:14}
  Under the conditions of Theorem~\ref{th:3} for every $b>0$
  \begin{equation}
    \label{eq:36}
    \int_0^1 \sup_{(v,x,q)\in C(b)}  \frac{1}{(\frac{\partial r}{\partial x}(v,x))^2}\left|
      v^m\frac{\partial H_t}{\partial v^m}(v,x,q)\right|^2 dt < \infty,
    \quad m=1,\dots,M,
  \end{equation}
  where
  \begin{displaymath}
    C(b) \set \descr{(v,x,q)\in \mathbf{A}}{\abs{q}\leq b}. 
  \end{displaymath}
\end{Lemma}
\begin{proof}
  For $(v,x,q) \in \mathbf{A}$ we already know from Lemma~\ref{lem:13}
  that
  \begin{displaymath} \frac{\partial H_t}{\partial v_m}(v,x,q)
    =\mathbb{E}[\frac{\partial^2 r}{\partial x\partial
      v_m}(v,\Sigma(x,q))(D_t \Sigma_0 +
    \ip{q}{D_t\psi})|\mathcal{F}_t].
  \end{displaymath} 
  According to Theorem~4.2 in \cite{BankKram:11a} there is a constant
  $b_1>0$ such that for all $(v,x)\in (0,\infty)^M\times \mathbf{R}$
  and $m=1,\dots,M$
  \begin{displaymath}
    \frac1{b_1} \frac{\partial r}{\partial x}(v,x) \leq v^m
    \frac{\partial^2 r}{\partial v^m \partial x}(v,x) \leq b_1
    \frac{\partial r}{\partial x}(v,x). 
  \end{displaymath}
  Accounting for~\eqref{eq:4} we deduce that for $(v,x,q) \in C(b)$
  \begin{displaymath}
    0 < v^m \frac{\partial^2 r}{\partial
      v^m \partial x}(v,\Sigma(x,q)) \leq b_1 \frac{\partial
      r}{\partial x}(v,\Sigma(x,q)) \leq b_1 \frac{\partial
      r}{\partial x}(v,x)
    e^{c(\Sigma_0^- + b\abs{\psi})/M}.
  \end{displaymath}
  It follows that
  \begin{align*}
    \abs{\frac{v^m}{\frac{\partial r}{\partial x}(v,x)} \frac{\partial
        H_t}{\partial v^m}(v,x,q)}^2 &\leq b_1^2
    \left(\mathbb{E}[e^{c(\Sigma_0^- + b\abs{\psi})/M}(\abs{D_t
        \Sigma_0} +
      b\abs{D_t\psi})|\mathcal{F}_t]\right)^2 \\
    &\leq b_1^2 \mathbb{E}[e^{2c(\Sigma_0^- +
      b\abs{\psi})/M}|\mathcal{F}_t] \; \mathbb{E}[(\abs{D_t \Sigma_0}
    + b\abs{D_t\psi})^2|\mathcal{F}_t],
  \end{align*}
  which proves~\eqref{eq:36} because, in view of~\eqref{eq:17}, the
  martingale $\mathbb{E}[e^{2c(\Sigma_0^- +
    b\abs{\psi})/M}|\mathcal{F}_t]$, $t \in [0,1]$ has bounded paths
  and because of~\eqref{eq:35}.
\end{proof}

The proof of Theorem~\ref{th:3} is completed.

\bibliographystyle{plainnat} \bibliography{../bib/finance}

\end{document}